\documentclass[final,12pt,pmlr]{jmlr} 

\jmlrvolume{}
\jmlryear{}
\jmlrpages{}
\jmlrproceedings{}{}
\makeatletter
\renewcommand{\@titlefoot}{}
\makeatother

\title[Lower Bounds for Differential Privacy Under Continual Observation]{Lower Bounds for Differential Privacy Under Continual Observation and Online Threshold Queries}
\usepackage{times}

\author{%
 \Name{Edith Cohen} \Email{edith@cohenwang.com}\\
 \addr Google Research and Tel Aviv University
 \AND
 \Name{Xin Lyu} \Email{lyuxin1999@gmail.com}\\
 \addr UC Berkeley and Google Research
 \AND
 \Name{Jelani Nelson} \Email{minilek@alum.mit.edu}\\
 \addr UC Berkeley and Google Research
 \AND
 \Name{Tam\'{a}s Sarl\'{o}s} \Email{stamas@google.com}\\
 \addr Google Research
 \AND
 \Name{Uri Stemmer} \Email{u@uri.co.il}\\
 \addr Tel Aviv University and Google Research
}

\newenvironment{ignore}[1]{}

\usepackage{amsfonts}
\usepackage{dsfont}
\usepackage{bbm}
\usepackage{cleveref}
\usepackage{xcolor}

\usepackage{amssymb}
\usepackage{enumitem}

\DeclareMathSymbol{\N}{\mathbin}{AMSb}{"4E}
\DeclareMathSymbol{\Z}{\mathbin}{AMSb}{"5A}
\DeclareMathSymbol{\R}{\mathbin}{AMSb}{"52}
\DeclareMathSymbol{\Q}{\mathbin}{AMSb}{"51}
\def\error{\mbox{\rm error}}
\newcommand{\calA}{\mathcal{A}}
\newcommand{\AAA}{\mathcal{A}}
\newcommand{\XXX}{\mathcal{X}}
\newcommand{\MMM}{\mathcal{M}}

\newcommand{\calB}{\mathcal{B}}
\newcommand{\BBB}{\mathcal{B}}
\newcommand{\Releases}{\texttt{Releases}}
\newcommand{\eps}{\varepsilon}

\def\Lap{\textsf{Lap}}
\def\Sim{\textsf{Sim}}
\def\Helper{\textsf{Helper}}
\def\Uniform{\textsf{Uniform}}
\def\Bern{\textsf{Bern}}
\def\Geom{\textsf{Geom}}
\def\NegBinomial{\textsf{NegBinomial}}
\def\supp{\mathsf{supp}}

\newtheorem{claim}[theorem]{Claim}

\newtheorem{observation}[theorem]{Observation}
\newtheorem{problem}[theorem]{Problem}

\newcommand{\edith}[1]{{\color{purple} Edith: #1}}

\newcommand{\remove}[1]{}

\begin{document}

\maketitle

\begin{abstract}
One of the most basic problems for studying the ``price of privacy over time'' is the so called {\em private counter problem}, introduced by Dwork et al.\ (2010) and Chan et al.\ (2011). 
In this problem, we aim to track the number of {\em events} that occur over time, while hiding the existence of every single event. More specifically, in every time step $t\in[T]$ we learn (in an online fashion) that $\Delta_t\geq 0$ new events have occurred, and must respond with an estimate $n_t\approx\sum_{j=1}^t \Delta_j$. The privacy requirement is that {\em all of the outputs together}, across all time steps, satisfy {\em event level} differential privacy. 

The main question here is how our error needs to depend on the total number of time steps $T$ and the total number of events $n$. Dwork et al.\ (2015) showed an upper bound of $O\left(\log(T)+\log^2(n)\right)$, and Henzinger et al.\ (2023) showed a lower bound of $\Omega\left(
\min\{\log n, \log T\}
\right)$. We show a new lower bound of $\Omega\left(\min\{n,\log T\}\right)$, which is tight w.r.t.\ the dependence on $T$, and is tight in the sparse case where $\log^2 n=O(\log T)$. Our lower bound has the following implications:
\begin{itemize}
    \item We show that our lower bound extends to the {\em online thresholds} problem, where the goal is to privately answer many ``quantile queries'' when these queries are presented one-by-one. This resolves an open question of Bun et al.\ (2017).
    \item Our lower bound implies, for the first time, a separation between the number of mistakes obtainable by a private online learner and a non-private online learner. This partially resolves a COLT'22 open question published by Sanyal and Ramponi.
    \item Our lower bound also yields the first separation between the standard model of private online learning and a recently proposed relaxed variant of it, called {\em private online prediction}.

\end{itemize}

\end{abstract}

\section{Introduction}
Differential privacy (DP), introduced by \cite{DMNS06}, is a mathematical definition for privacy that aims to enable statistical analysis of datasets while providing strong guarantees that individual-level information does not leak. Most of the academic research on differential privacy has focused on the ``offline'' setting, where we analyze a dataset, release the outcome (in a privacy preserving manner), and conclude the process. In contrast, most of the industrial applications of differential privacy are ``continual'' and involve re-collecting and re-aggregating users' data, say on a daily basis. 
In this work we study one of the most popular models for arguing about this, called the {\em continual observation model} of differential privacy \citep{DBLP:conf/focs/DworkRV10,chan2011private}. Algorithms in this model receive a stream of sensitive inputs and update their outcome as the stream progresses while guaranteeing differential privacy. This model abstracts situations where the published statistics need to be continually updated while sensitive information is being gathered. 

Before describing our new results, we define our setting more precisely. Consider an algorithm $\AAA$ that on every time step $t\in[T]$ obtains a (potentially empty) dataset $D_t\in X^*$ and outputs an answer $y_t$. We say that $\AAA$ solves a problem $P:X^*\rightarrow\R$ if for every sequence of datasets $D=(D_1,\dots,D_T)$, with high probability over $\AAA$'s coins, for every $t\in[T]$ it holds that $y_t\approx P(D_1\circ\dots\circ D_t)\triangleq p_t$.
We measure quality by the $\ell_\infty$ norm, that is, for a sequence of outputs $y=(y_1,\dots,y_T)$ we denote  $\error_P(D,y)=\max_{t\in[T]}\left|y_t-p_t\right|$.
The privacy requirement is that the {\em sequence of all $T$ outputs} should be differentially private w.r.t.\ $D$.
That is, the outcome distribution of the $y$ should be insensitive to a change limited to one elements in one of the datasets $D_t$. 
Formally,

\begin{definition}[\cite{DMNS06}]\label{def:DP}
Given a sequence of datasets $D=(D_1,\dots,D_T)$ and an algorithm $\AAA$, we write $\Releases_{\AAA}(D)$ to denote the sequence of $T$ outcomes returned by $\AAA$ during the execution on $D$. Algorithm $\calA$ is {\em $(\eps,\delta)$-differentially private} if for any {\em neighboring}\footnote{$D=(D_1,\dots,D_T)$ and $D'=(D'_1,\dots,D'_T)$ are neighboring if $D_t=D'_t$ for all but at most index $t^*$ for which $D_{t^*}$ could be obtained from $D'_{t^*}$ by adding/removing one element.} sequences of datasets $D,D'$, and for any outcome set $F\subseteq \R^T$ it holds that
$
\Pr[\Releases_{\AAA}(D)\in F]\leq e^{\eps}\cdot\Pr[\Releases_{\AAA}(D')\in F]+\delta.
$  The case where $\delta=0$ is referred to as \emph{pure} DP and the case where $\delta>0$ as \emph{approximate} DP.
\end{definition}

\begin{remark}
Typically one aims for a stronger privacy definition in which the neighboring datasets are not fixed in advanced, and can be adaptively chosen as the execution progresses. Since in this work we show {\em negative} results, then focusing on the above ``non-adaptive'' privacy definition only makes our results stronger.
\end{remark}

An algorithm that operates in the continual observation model must overcome two main challenges: (i) The algorithm is required to solve $T$ instances of the underlying problem $P$, rather than just a single instance, and (ii) The algorithm receives its input in an online manner, which may introduce larger error than what is required when the algorithm receives the entire sequence at once. So far, all known hardness results for the continual observation model were based on challenge (i) above. To emphasize this, let us consider the following ``offline'' variant of the continual observation model:

\begin{definition}[The offline observation model] Let $P:X^*\rightarrow\R$ be a problem.
In the {\em offline observation model}, the algorithm gets in the beginning of the execution an entire sequence of datasets $D=(D_1,\dots,D_T)$, and aims to return $T$ solutions $y_1,...,y_T$ such that $y_t\approx P(D_1\circ\dots\circ D_t)$ for every $t\in[T]$.
\end{definition}

So far, all prior hardness results for the continual observation model are, in fact, hardness results for this offline observation model. 
On the one hand, this makes these prior results stronger, as the offline observation model is strictly easier than the continual observation model. On the other hand, such lower bounds cannot capture the hardness that steams from the online nature of the continual observation model. As we show in this work, there are important cases where in order to show strong lower bounds for the continual observation model we {\em must} differentiate it from the offline observation model.\footnote{We remark that in other models of computation there are known hardness result that differentiate between similar offline and online models, most notably by \cite{BunSU:SODA2017}.}
Furthermore, we show this for what is arguably the most well-studied problem in the continual observation model -- the {\em private counter} problem. In this problem, the data domain is $X=\{0,1\}$ and the function $P$ simply counts the number of ones in the datasets received so far. Alternatively, it is common to interpret the input as a sequence of non-negative numbers $\Delta=(\Delta_1,\dots,\Delta_T)\in \mathbb{Z}_{\geq 0}^T$, where $\Delta_t$ is the number of {\em events} that occurred at time $t$, and the goal is to track (and publish) the prefix sums $n_t = \sum_{i=1}^t \Delta_i$ for $t\in [T]$.

\subsection{Prior Works on the Counter Problem}

\paragraph{Upper bounds.}
\cite{DworkNPR10} and \cite{chan2011private} presented the celebrated \emph{binary tree mechanism} for the online counter problem, which guarantees error at most
$O_{\eps}(\log(T) + \log^{2.5} n)$, where $n=\sum_{t=1}^T \Delta_t$ is the total number of events throughout the execution (via the analysis of \cite{asiacrypt-2015-27369}). In the offline observation model, the private counter problem can be solved with error $\tilde{O}_{\eps,\delta}((\log^* T) ( \mathrm{polylog}(n)))$ via algorithms for the related problem of {\em threshold sanitization} (a.k.a.\ CDF estimation) \citep{CLNSS:STOC2023,KaplanLMNS20}. So, in terms of the dependence in the time horizon $T$, the private counter problem can be solved with error $O(\log T)$ in the online model and $\tilde{O}(\log^* T)$ in the offline model.

\paragraph{Lower bounds.}
There are three prior lower bounds for the private counter problem, all of which hold also for the offline model: 
\begin{enumerate}[itemsep=0px]
    \item[(1)] \cite{DworkNPR10} showed a lower bound of $\Omega(\log T)$ that holds only for \emph{pure}-DP algorithms.
    \item[(2)] For approximate DP, a lower bound of $\Omega(\log^* T)$ follows from the results of~\cite{DBLP:conf/focs/BunNSV15}, who showed a lower bound for the simpler {\em private median} problem. Thus, in the offline setting, the dependence on $T$ of $\tilde{\Theta}(\log^* T)$ is tight.
    \item[(3)] \cite{HenzingerUU:SODA2023} recently established a lower bound of $\Omega(\min\{\log T, \log n\})$ on the $\ell_2$ error, which implies a lower bound of $\Omega(\min\{\log T,\log n\})$ on the $\ell_\infty$ error. 
\end{enumerate}

Our focus in this work is on lower bounds on the $\ell_\infty$ error of approximate-DP online counters.  To summarize our knowledge from prior works: The known upper bound is $O(\min\{n,\log T +\log^2 n\})$, combining the bound of~\cite{asiacrypt-2015-27369} with the observation that we can trivially obtain an error of $n$ by reporting $\tilde{n}_t = 0$ for $t\in [T]$). The prior lower bound is
$\Omega(\min\{\log T, \log n\})$ by \cite
{HenzingerUU:SODA2023}. Note that there is a gap: In the common regime where $n\leq T$ the upper bound is $O(\log T + \log^2n)$ while the lower bound is only $\Omega(\log n)$. Furthermore, as the lower bound of~\cite{HenzingerUU:SODA2023} applies in the offline setting, its dependence on $T$ cannot be tightened, since the offline setting has an upper bound of $\tilde{O}((\log^* T) \mathrm{polylog}(n))$. 

\begin{remark}
Most applications for the counter problem operate in the sparse regime where $n\ll T$. For example, in tracking the view counts of YouTube videos, popular videos like those of Taylor Swift may experience rapid increases in views, while the majority of videos tend to exhibit slower growth (that is, small $n$). This typical pattern is broadly known as the long-tail phenomena.
In this use case, a counter is maintained for each video. We may need frequent refreshes (e.g. large $T$), for example, so that we can detect sudden surges quickly and verify the content for compliance. 
\end{remark}

\subsection{Our Contributions}

We present a lower bound of $\Omega(\min\{n,\log T\})$ on the $\ell_\infty$ error of approximate-DP online counting. This bound is tight in terms of the dependence on $T$ and is tight overall for sparse inputs,
where $\log^2 n = O(\log T)$. Formally,
\begin{theorem} \label{counterLB:thm}
    For any $T\ge 1$, let $n \leq \frac{1}{2}\log(T)$ and $\delta \leq \frac{1}{100n}$. Any $(0.1, \delta)$-DP algorithm for the online counter problem with $n$ events over $T$ time steps must incur $\ell_\infty$-error of $\Omega(n)$ with constant probability.
\end{theorem}

Observe that this gives a lower bound of $\Omega(\log T)$ even when $n$ is as small as $\log T$. As we explained, such a lower bound does not hold for the offline model, and hence our result differentiates between these two models.

\paragraph{The threshold monitor problem.} To obtain our results, we consider a simpler variant of the counter problem, which we call the {\em Threshold Monitor Problem}, which is natural and of independent interest. Similarly to the counter problem, on every time step $t$ we obtain an input $\Delta_t$ indicating the number of events occurring on step $t$. However, instead of maintaining a counter for the total number of events, all we need to do is to indicate when it (roughly) crosses some predefined  threshold $k$. Formally,

\begin{problem} [The threshold monitor problem] \label{threshmon:problem}
There is a pre-determined threshold value $k$.
The input to the algorithm is a sequence of $T$ updates $(\Delta_1,\dots,\Delta_T)\in \{0,1\}^T$. At each time step $t\in [T]$, the algorithm receives the update $\Delta_t$, lets $n_t \gets \sum_{j=1}^{t} \Delta_j$, and outputs $\top$ or $\bot$, subject to the following requirements:
\begin{itemize}[itemsep=0px,topsep=4px]
    \item If $n_t < \frac{1}{2}k$, the algorithm reports $\bot$.
    \item If $n_t \ge k$, the algorithm reports $\top$.
    \item If $n_t \in [\frac{1}{2}k, k)$, the algorithm may report either $\bot$ or $\top$.
    \item Once the algorithm reports a $\top$, it halts.
\end{itemize}
We say the algorithm succeeds if all its responses (before halting) satisfy the requirements above.
\end{problem}

Clearly, an algorithm for the counter problem with $\ell_\infty$ error at most $k/2$ can be used to solve the Threshold Monitor Problem. Therefore, a lower bound for the Threshold Monitor Problem implies a lower bound for the counter problem. We establish the following, which implies Theorem~\ref{counterLB:thm}. The proof is included in Section~\ref{monitorLBproof:sec}.
\begin{theorem} \label{monitorLB:thm}
    Consider the Threshold Monitor Problem with $T\ge 1$, $k \leq \frac{1}{2}\log(T)$, and privacy parameter $\delta \leq \frac{1}{100k}$. There is no $(0.1,\delta)$-DP algorithm that succeeds with probability $0.99$.
\end{theorem}

A natural question is whether 
the requirement of $\delta < \frac{1}{\log(T)}$ can be relaxed to $\delta < o(1)$. The following observation suggests that the answer is negative. 
\begin{observation}
    There is an $(\eps,\delta)$-DP algorithm for the Threshold Monitor Problem with threshold $k=O\left(\min\left\{\frac{\log T}{\eps} , \frac{1}{\delta}\right\}\right)$ that succeeds with constant probability.
\end{observation}
This follows as a special case of the classical \texttt{AboveThreshold} algorithm of~\cite{DworkNRRV09}, also known as the \emph{Sparse Vector Technique}. The algorithm performs an online sequence of $T$ noisy threshold tests and halts after the first time that the response is that the threshold is exceeded. The overall privacy cost is that of a single query. For the threshold monitor problem, we test sequentially each prefix sum if it is above the threshold $3k/4$. We return $\bot$ on ``below'' responses and return $\top$ (and halt) on the first ``above'' response. We require that the $\ell_\infty$ norm of the noise does not exceed $k/4$ (w.h.p.), which holds when $k= \Omega\left(\frac{\log T}{\eps}\right)$. The $O\left(\frac{1}{\delta}\right)$ upper bound follows by simply releasing every input bit with probability $\delta$, and estimating the sum of the bits from that.

\subsubsection{Application I: The threshold queries problem}

We show that our lower bound extends to the {\em threshold queries problem}, thereby solving an open question of~\cite{BunSU:SODA2017}. In this problem, we get an input dataset $D$ containing $n$ points from $[0,1]$. Then, for $T$ rounds, we obtain a query $q\in[0,1]$ and need to respond with an approximation for the number of points in $D$ that are smaller than $q$. The privacy requirement is that the sequence of all our answers together satisfies DP w.r.t.\ $D$. Note that here the dataset is fixed at the beginning of the execution and it is the {\em queries} that arrive sequentially. 

In the offline variant of this problem, where all the queries are given to us in the beginning of the execution (together with the dataset $D$), this problem is equivalent to the offline counter problem, and can be solved with error $\tilde{O}(\log^* T)$ using the algorithm of~\cite{CLNSS:STOC2023,KaplanLMNS20}. For the online variant of the problem (where the queries arrive sequentially), \cite{BunSU:SODA2017} presented an upper bound of $O(\log T)$, and asked if it can be improved. We answer this question in the negative and show the following theorem.

\begin{theorem}\label{thm:threshld-query}
Let $n,T\ge 1$ be two integers. There is no $(0.1,\frac{1}{n^{2}})$-DP algorithm $\calA$ that, on input a set $S$ of $n$ reals number from $[0,1]$, answers a sequence of $T$ online but non-adaptive threshold queries within error $\frac{1}{2}\min\{n,\log(T)\}$ and with probability $0.99$. 
\end{theorem}

\subsubsection{Application II: Private online learning}

We apply our technique to establish a lower bound for the \emph{private online learning} problem in the \emph{mistake bound model}. To the best of our knowledge this is the first non-trivial lower bound for the private online learning model, separating it from the non-private online learning model. This partially answers a COLT'22 open question posted by \cite{sanyal2022you}.\footnote{This does not completely resolves their question as they aimed for a stronger separation.} 
Our lower bound also implies, for the first time, a separation between the problems of {\em private online learning} and \emph{private online prediction}. Let us recall these two related models.

\paragraph{The private online learning model.} 
Suppose $\mathcal{X}$ is the data domain and Let $\mathcal{H}\subseteq \{0,1\}^{\mathcal{X}}$ be a hypothesis class. The private online learning model, considered first by 
\cite{GolowichL21-littlestone},
can be expressed as a $T$-round game between an {\em adversary} $\BBB$ and a learning algorithm (the learner) $\AAA$. In every round $i$ of this game:
\begin{enumerate}[itemsep=0px]
    \item The adversary chooses a query $x_i\in\mathcal{X}$ and the learner  chooses a hypothesis $h_i:\XXX\rightarrow\{0,1\}$.
    \item At the same time, the adversary gets $h_i$ and the learner gets $x_i$.
    \item The adversary chooses a ``true'' label $y_i$ under the restriction that $\exists h\in\mathcal{H}$ such that $\forall j\leq i$ we have $h(x_j)=y_j$. This ``true'' label $y_i$ is given to the learner.
\end{enumerate}

We say that a mistake happened in round $i$ if $h_i(x_i) \ne y_i$. The goal of the learner is to minimize the total number of mistakes. The privacy definition is that the sequence of all hypotheses $(h_1,h_2,\dots,h_T)$ together satisfies differential privacy w.r.t.\ the inputs given to the algorithm throughout the execution. That is, if we were to replace one labeled example $(x_i,y_i)$ with another $(x'_i,y'_i)$ then that should not change by much the distribution on the resulting vector of hypotheses.

\begin{remark}
In Step~1 of the game above, the adversary and the learner chooses the next query and hypothesis {\em simultaneously}. There is another formulation of the model in which the learner is allowed to choose the current hypothesis $h_i$ as a function of the current query $x_i$. Note that, in both models, $h_i(x_i)$ clearly must depend on $x_i$. The question is if the description of the hypothesis itself is allowed to depend on $x_i$ (in a way that preserves DP). We do not know if our lower bound holds also for this alternative formulation of the model.
\end{remark}

\paragraph{The private online prediction model.}
This is a similar model, introduced by \cite{KaplanMMNS23-prediction}, which incorporates the following two changes: 
First, in Step~2 of the above game, instead of getting the hypothesis $h_i$, the adversary only gets $h_i(x_i)$. Second, the privacy requirement is weaker. Specifically, instead of requiring that the vector of {\em all} outputs given by the learner should be insensitive to changing one labeled example $(x_i,y_i)$, the requirement is only w.r.t.\ the vector of all outputs {\em other than the output in round $i$}. 
Intuitively, in the context of private online prediction, this privacy definition allows the $i$th prediction to depend strongly on the $i$th input, but only very weakly on all other inputs. This privacy notion could be a good fit, for example, when using a (privacy preserving) online learner to predict a medical condition given an individual's medical history: As long as the $i$th individual does not share its medical data with additional users, only with the learner, then its medical data will not leak to the other users via the predictions they receive.  \cite{KaplanMMNS23-prediction} showed strong positive results that outperform the known constructions for the standard model of private online learning. However, before our work, it was conceivable that the two models are equivalent, and no separation was known.

\medskip

We consider the class of \emph{point functions} over a domain $\XXX$, containing all functions that evaluate to $1$ on exactly one point of the domain $\XXX$.\footnote{That is, for every $x\in\XXX$, the class of point functions contains the function $c_x$ defined as $c_x(y)=1$ iff $y=x$.}
This class is known to have Littlestone dimension 1, and hence there is a non-private online-learner for it that makes at most $1$ mistake throughout the execution (no matter what $T$ is). We apply our techniques to establish a non-trivial lower bound for the private online learning model, thereby separating it from the non-private model. This is captured by the following theorem.

\begin{theorem}\label{thm:onlineLearning}
    Let $T\ge 100$ denote the number of time steps, and let $\mathcal{X}=\{0,1,2,\dots,T\}$ be the data domain. Let $\mathcal{H}$ be the class of {\em point functions} over $\mathcal{X}$. Let $\delta = \frac{1}{20\log(T)}$. Then, for any $(\eps,\delta)$-DP online learner $\AAA$ for $\mathcal{X}$ that answers $T$ queries there exists a query sequence on which $\AAA$ makes $\Omega(\log(T))$ mistakes with constant probability. 
\end{theorem}

We complement this result with an upper bound for the online private prediction model that achieves a mistake bound which is independent of $T$. This establishes a separation between private online learning and online prediction. Formally,

\begin{lemma}\label{lem:onlinePredictionPoints}
    Let $\mathcal{X}$ be a domain. The class of point functions  over $\mathcal{X}$ admits an $(\eps, \delta)$-DP online predictor with mistake bound $O(\frac{1}{\eps^2}\log^2\frac{1}{\delta})$.
\end{lemma}

The proofs of Theorem~\ref{thm:onlineLearning} is given in Appendix~\ref{sec:lowerboundOnlineLearning} and the proof of Lemma~\ref{lem:onlinePredictionPoints} is given in Appendices~\ref{sec:JDPmirror} and~\ref{JDPpoint:sec}.
In Appendix~\ref{sec:JDPmirror} we present an algorithm that satisfies \emph{Joint Differential Privacy}~\citep{KearnsPRRU15} for a version of the threshold monitor problem that applies with 
$k=O(\frac{1}{\eps^2}\log^2\frac{1}{\delta})$ 
and in Appendix~\ref{JDPpoint:sec} we apply it to construct an online prediction algorithm.


\section{A Lower Bound for the Threshold Monitor Problem} \label{monitorLBproof:sec}

In this section we prove our lower bound for the threshold monitor problem (i.e., we prove \Cref{monitorLB:thm}), which directly implies our lower bound for the counter problem. We begin with an overview of the proof.

\subsection{Proof Overview}
Recall that in the threshold monitor problem, the algorithm receives a sequence of updates $\Delta_1,\dots, \Delta_T \in \{0,1\}^T$ and needs to report ``HALT'' at some time step $t\in [T]$ such that $\sum_{i=1}^{t} \Delta_i \in [k/2,k]$, where $k = \frac{1}{2}\log(T)$. 
We want to show that no $(\eps = 0.1, \delta = \frac{1}{100k})$-DP algorithm can solve the problem with success probability $0.99$. Suppose such an algorithm exists and denote it by $\calA$. Then, on the sequence $D^{(0)} = (0,0,\dots, 0)$, $\calA$ cannot halt with probability larger than $0.01$.

To derive a contradiction, we want to use the fact that $\calA$ is differentially private to find a new data set $D^{(k)} \in \{0,1\}^T$ that contains $k$ ``1''s in the sequence, yet the algorithm halts on $D^{(k)}$ with probability at most $0.5$. 
If we just append $k$ ``1'' at the beginning of the sequence, the DP property of $\calA$ only promises that $\calA$ halts with probability at most:
$0.01 e^{\eps k} + \sum_{i=1}^{k} \delta \cdot e^{(k-i) \eps}$,
which is a vacous bound for $\delta = \frac{1}{100k}$.

In this naive attempt, we only considered the probability that $\calA$ ever halts on the input. By the DP property of the algorithm, the best we can say about this quantity is that it increases by a multiplicative factor of $e^{\eps}$ and an additive factor of $\delta$ when we move to an adjacent data set. The additive increment is easy to bound. However, to handle the multiplicative blow-up, we need to exploit the online nature of the problem and the long time horizon. 
Specifically, on an input sequence $S = (\Delta_1,\dots, \Delta_T)\in \{0,1\}^{T}$, we consider whether $\calA$ is more likely to halt on the stream's first or second half.

\begin{itemize}
    \item If $\calA$ more likely halts at time step $t\in [1, T/2]$, then we fix the first half of $S$, and recursively insert more ``$1$''s in the second half of $S$. Because $\calA$ is an online algorithm, only the second half of the stream will suffer from the ``$e^{\eps}$'' blow-up in the future.
    \item If $\calA$ more likely halts at time step $t\in [T/2+1,T]$, then we fix the second half of $S$, and recursively insert more $1$'s in the first half of $S$. In doing so, the behavior of $\calA$ in time interval $[T/2+1,T]$ might change. However, our construction will eventually append at least $k$ ``$1$''s before time step $T/2$. Therefore, ultimately, the algorithm's behavior on the second half of the stream is irrelevant.
\end{itemize}
Using this construction, we can start at $D^{(0)}= (0,\dots, 0)^{T}$, find a sequence $D^{(k)}$ as well as an index $\ell$, such that $D^{(k)}$ contains at least $k$ ``$1$''s before time step $\ell$, yet $\calA$ halts before time step $\ell$ with probability at most $0.5$. The key insight in the construction is intuitively explained as follows: after inserting each new item, we cut the stream into two halves and recurse into the left or right half, in a way that at most half of the probability mass will suffer from the $e^{\eps}$ blow-up in the future. Overall, the $e^{\eps}$ blow-up is reduced to a $\frac{1}{2}e^{\eps}$ multiplicative factor, which is less than $1$, and the overall halting probability is easily bounded.

\subsection{The Formal Details}
We now present the formal construction. We prove the following restatement of Theorem~\ref{monitorLB:thm}.
\begin{theorem}
   Let $T > 1$ and $k \le \frac{1}{2}\log(T)-5$.
   There is no $(\frac{1}{2}, \frac{1}{100\cdot k})$-DP algorithm for the Threshold Monitor Problem (Problem~\ref{threshmon:problem}) with parameters $T,k$ that has success probability $0.99$.
\end{theorem}

\begin{proof}
Fix $k$ and let $T = \sum_{i=1}^{k} 2^i$. It is clear that $k\le \frac{1}{2}\log(T)-5$. We show that there is no algorithm for the Threshold Monitor Problem with parameters  $k$ and $T$ (or more) updates. 
We assume that such an algorithm $\calA$ exists and derive a contradiction in the following.

For any fixed input sequence $D \in \{0,1\}^{T}$ and $t\in[T]$, 
an algorithm $\calA$ for Problem~\ref{threshmon:problem} can be fully described by a probability distribution on $[T]\cup \{\textrm{non}\}$ that corresponds to the time step $t\in [T]$ where it outputs $\top$ and halts. The event $\textrm{non}$ corresponds to the algorithm outputting $\bot$ on all time steps.  We use the notation $p^{D}_{t}$ for
the probability that $\calA$ outputs $\top$ (and thus halts) after the $t$-th update:
\begin{equation} \label{probmapA:eq}
p^{D}_{i} := \Pr[\calA(D) \text{ outputs $\top$ after the $i$-th update}].
\end{equation}

We start with the all-zero sequence $D^{(0)} = (0,0,\dots)$ and construct a ``hard sequence'' for algorithm $\calA$, as described in Algorithm~\ref{algo:construction}.

\newcommand{\ctotal}{c_{\mathrm{total}}}

\begin{algorithm2e}[H] \label{hardseq}
{\small
\LinesNumbered
    \caption{Construction of Hard Instance}
    \label{algo:construction}
    
    \DontPrintSemicolon
    
    \KwIn{
         Parameters $k$, $T = \sum_{j=1}^{k} 2^j$; A mapping $p^{D}_{i}$ as in \eqref{probmapA:eq} induced by an Algorithm $\calA:\{0,1\}^T$. 
    }
    Initialize $\ell \gets 1$, $r\gets T$, $D^{(0)}\gets (0,0,\dots, 0)$, and $c_\mathrm{total} \gets 0$ \;
    \For{$i = 1,\dots, k$}{
        $m\gets \frac{1}{2}(\ell + r + 1)$ \tcp*{midpoint $m$ is always an integer} 
        $c_{\textrm{left}} \gets \sum_{j=\ell}^{m-1} p^{D^{(i-1)}}_j$ \;
        $c_{\textrm{right}} \gets \sum_{j=m}^{r} p^{D^{(i-1)}}_j$ \;
        \If{$c_{\mathrm{left}} \le c_{\mathrm{right}}$}{
            Set the $\ell$-th entry of $D^{(i-1)}$ to $1$ to obtain $D^{(i)}$ \;
            $\ctotal \gets \ctotal + p_{\ell}^{D^{(i)}}$ \;
            $\ell\gets \ell + 1$ \;
            $r\gets m - 1$ \; 
        }
        \Else{
            Set the $m$-th entry of $D^{(i-1)}$ to $1$ to obtain $D^{(i)}$ \;
            $c_{\mathrm{total}} \gets \ctotal + c_{\mathrm{left}} + p^{D^{(i)}}_{p}$ \;
            $\ell \gets m + 1$ \;
        }
    }
    \KwRet{$\ell, \ctotal$ and $D^{(k)}$}
    }
\end{algorithm2e}

Note that $D^{(k)}$ is an input sequence with $k$ ``$1$''s before the $\ell$-th update. Thus, a correct algorithm has to output $\top$ before the $\ell$-th update arrives. However, we will show that, assuming the algorithm is differentially private and is correct on the first sequence $D^{(0)}$ with probability $0.99$, the algorithm must err on the input $D^{(k)}$ with probability at least $\frac{1}{2}$, contradicting to the correctness of the algorithm. 
We now give the proof. We start with the following simple observation.

\begin{claim}
    Consider \Cref{algo:construction}. After the algorithm finishes, it holds that
    \[
    \ctotal = \Pr[\calA(D^{(k)}) \text{ outputs $\top$ before or at the $\ell$-th update}].
    \]
\end{claim}

\begin{proof}
By definition, we have
\[
\Pr[\calA(D^{(k)}) \text{ outputs $\top$ before or at the $\ell$-th update}] = \sum_{i=j}^r p^{D^{(k)}}_j.
\]
On the other hand, considering the execution of \Cref{algo:construction}, we observe that
$
\ctotal = \sum_{j=1}^{\ell} p^{D^{(k_j)}}_j,
$ 
where we use $k_j$ to denote the round in which the interval $[\ell, r]$ in \Cref{algo:construction} updates to a new interval $[\ell', r']$ with $\ell' > j$. Observe that $p^{D^{(k_j)}}_j = p^{D^{(k)}}_j$
for every $j$. (This follows because after the $k_j$-th round, we only update the $\ell'$-th to the $r'$-th entries of the input, and the algorithm is ignorant to future updates at time $j$.) The claim now follows.
\end{proof}

The next claim gives an upper bound for $\ctotal$.

\begin{claim}
    Consider \Cref{algo:construction}. After the execution finishes, we have 
    $
    \ctotal \le \frac{e^{\eps}\left(0.01 + 2 k \delta\right)}{1 - \frac{1}{2} e^\varepsilon}.
    $
\end{claim}

\begin{proof}
    We shall use a potential argument. Before the $i$-th round of \Cref{algo:construction}, let the current interval be $[\ell, r]$ and define the current \emph{potential} of the algorithm by
    $
    \Phi^{(i-1)} = \sum_{j=\ell}^r p^{D^{(i-1)}}_{j}.
    $
    Let $\ctotal^{(i-1)}$ be the value of $\ctotal$ before the $i$-th round. We will examine (with some foresight) the following quantity
    $
    e^{-\eps} \ctotal^{(i)} + \frac{1}{1-\frac{1}{2}e^{\eps}} \Phi^{(i)}.
    $
    First, assuming the algorithm errs with probability at most $0.01$ on the input $D^{(0)}$, it follows that
    \[
    e^{-\eps} \ctotal^{(0)} + \frac{1}{1-\frac{1}{2}e^{\eps}} \Phi^{(0)} \le 0 + \frac{0.01}{1-\frac{1}{2}e^\eps}.
    \]
    Consider how $\ctotal$ and $\Phi$ change from round $i-1$ to round $i$. Note that $\ctotal^{(i)}$ is obtained by adding a subset $I\subset [\ell, r]$ of $p^{D^{(i)}}_j$ to $\ctotal^{(i-1)}$. Assuming the algorithm is $(\eps,\delta)$-DP, we have
    \[
    \begin{aligned}
    \ctotal^{(i)} &= \ctotal^{(i-1)} + \sum_{j\in I} p^{D^{(i)}}_j 
    \le \ctotal^{(i-1)} + (e^{\eps} \sum_{j\in I} p^{D^{(i-1)}}_j + \delta) 
    \le \ctotal^{(i-1)} + e^{\eps} \Phi^{(i-1)} + \delta.
    \end{aligned}
    \]
    On the other hand, $\Phi^{(i)}$ consists of contribution from a sub-interval $J$ of $\Phi^{(i-1)}$, where the interval has a total contribution being at most half of $\Phi^{(i-1)}$. Again assuming the algorithm is $(\eps,\delta)$-DP, we have
    $
    \Phi^{(i)} 
    = \sum_{j\in J} p^{D^{(i)}}_j 
    \le \delta + e^\eps \cdot \sum_{j\in J} p^{D^{(i-1)}}_j 
    \le \frac{1}{2} e^\eps \Phi^{(i-1)} + \delta.
    $
    Thus,
    \[
    \begin{aligned}
    e^{-\eps} \ctotal^{(i)} + \frac{1}{1-\frac{1}{2}e^{\eps}} \Phi^{(i)}
    &\le e^{-\eps} \left( \ctotal^{(i-1)} + e^{\eps} \cdot \Phi^{i-1} +\delta \right) + \frac{1}{1-\frac{1}{2}e^\eps} \left( \frac{1}{2} e^\eps \Phi^{(i-1)} + \delta \right) \\
    &= \left( e^{-\eps} \ctotal^{(i-1)} + \frac{1}{1-\frac{1}{2}e^{\eps}} \Phi^{(i-1)} \right) + \frac{2\delta}{1-\frac{1}{2}e^{\eps}}.
    \end{aligned}
    \]
    Consequently,
    $
    e^{-\eps} \ctotal^{(k)} + \frac{1}{1-\frac{1}{2}e^{\eps}} \Phi^{(k)} \le \frac{0.01 + 2k\delta}{1-\frac{1}{2}e^\eps},
    $
    which implies the desired conclusion $\ctotal^{(k)} \le \frac{e^{\eps}(0.01 + 2k\delta)}{1-\frac{1}{2}e^\eps}$ as $\Phi$ is always non-negative.
\end{proof}

For our choice of $\eps = \frac{1}{2}$ and $\delta = \frac{1}{100 k}$, we have $\frac{e^{\eps}(0.01 + 2k\delta)}{1-\frac{1}{2}e^\eps} \le \frac{1}{2}$. Hence, combining the two claims above shows that the algorithm errs on the input $D^{(k)}$ with a probability of at least $0.5$, completing the proof.
\end{proof}

\section{Online Threshold Queries}

Recall that in the {\em threshold queries problem}, we get an input dataset $D$ containing $n$ points from $[0,1]$. Then, for $T$ rounds, we obtain a query $q\in[0,1]$ and need to respond with an approximation for the number of points in $D$ that are smaller than $q$. In this section we show that our lower bound extend also to this problem. We begin with a proof overview.

\subsection{Proof Overview}
The key challenge in extending our lower bound to the online threshold query problem is that the data set is given all at once. To motivate, let us consider a simple reduction:
\begin{itemize}
    \item For an input sequence $(\Delta_1,\dots, \Delta_T) \in \{0,1\}^{T}$ to the counter problem, construct a data set $D = \{\frac{i}{T+1}: \Delta_i = 1\}$. Then, issue a set of $T$ queries $\{q_i = \frac{i}{T}\}_{1\le i\le T}$ to the threshold query algorithm and report the answers to these queries.
\end{itemize}
For the offline version of the problem, this reduction preserves the privacy parameters and the correctness perfectly. However, this reduction itself does not work in the online setting: at any time step $t\in [1, T]$, the algorithm does not know the future inputs $\Delta_{t+1},\dots, \Delta_{T}$. Therefore, the data set $D$ cannot be constructed at all without seeing the whole stream.

To carry out the reduction for the online version of the problem, we only need one more property for the threshold query algorithm: on a query $q\in [0,1]$, the algorithm's response only depends on the set $D \cap [0,q]$. Namely, it is ignorant of all the data points with values larger than $q$. This condition is not necessarily satisfied by an arbitrary algorithm. Still, for any fixed algorithm $\calA$, it is possible to identify a set of $T$ points $0 < p_1,\dots, p_T < 1$ (possibly depending on $\calA$), such that the condition \emph{is} satisfied for any data set $D\subseteq \{p_1,\dots, p_T\}$ that only contains a subset of these $T$ points. The existence of such $T$ points is guaranteed by the hypergraph Ramsey theory. The high-level idea is easy to describe: Knowing the set $S = D\cap [0,q]$, for every possible data set $D$ that contains $S$, the algorithm's behavior on $q$ can be thought of as a ``coloring'' to the data set $D$. Since there are infinitely many points in $(q,1]$, one can use the Ramsey theory to identify an arbitrarily large monochromatic clique among $(q,1]$. That is, a set of data points $P\subseteq (0,q]$ such that the algorithm behaves similarly on any data set $D = S\cup P'$ with $P'\subseteq P$. Note that this is exactly what our reduction asks for.

\subsection{The Formal Details}
We are now ready to prove \Cref{thm:threshld-query}. We will do so by proving lower bounds for the following threshold version of the problem.

\begin{itemize}[itemsep=0px,topsep=4px]
    \item The input $D$ consists of $k$ data points, each taking values from $[0,1]$. The algorithm receives all of the set $D$ in advance.
    \item The algorithm then processes a sequence of queries $q_1\le q_2\le \dots \le q_T$ where each query $q_i\in [0,1]$ is described by a real number. Let $c_i = |\{j: x_j\le q_i\}|$.
    \begin{itemize}
        \item If $c_i \le \frac{1}{2}k$, the algorithm reports $\perp$.
        \item If $c_i \ge k$, the algorithm reports $\top$.
        \item If $c_i \in (k/2,k)$, the algorithm may report either $\perp$ or $\top$.
        \item Once the algorithm reports a $\top$, it halts.
    \end{itemize}
\end{itemize}
We say the algorithm succeeds if all the queries before halting are answered correctly. 
For $k = \frac{1}{2}\log(T)$ and $\eps = 0.1, \delta = \frac{1}{100k}$, we show that no $(\eps,\delta)$-DP algorithm can solve the problem with success probability more than $0.99$.

\subsubsection{Ramsey Theory}

We need a well-known result from Ramsey theory. Consider the complete $k$-uniform hypergraph over vertex set $V$ (that is to say, for every subset of size $k$, there is a hyperedge connecting them). For any integer $c\ge 1$, a $c$-coloring of the graph is a mapping $f:\binom{V}{k}\to [c]$ which assigns a color to each hyperedge. Given a mapping $f$, a subset $S\subseteq V$ of vertices is called a \emph{monochromatic clique} if all hyperedges inside $S$ are colored using the same color. 
Then, the hypergraph Ramsey theorem says the following.
\begin{lemma}
    Fix any finite $k,c, n\ge 1$. Then, there is a sufficiently large $N$ such that, for every $c$-coloring to a complete $k$-uniform hypergraph with at least $N$ vertices, a monochromatic clique of size at least $n$ exists.
\end{lemma}

\subsubsection{The Key Lemma}


We need the following technical lemma for the reduction to go through.

\begin{lemma}\label{lemma:apply-ramsey}
    Let $\calA$ be an algorithm for the online threshold problem that operates on $k$ data points. Then, for any $\eta > 0$, there exists a subset $0 < x_1 < \dots < x_T$ of $T$ points satisfying the following. For any fixed data set $S\subseteq \{x_1,\dots, x_T\}, |S| = k$ and the sequence of queries $x_1,\dots, x_T$, there exists a sequence of real numbers $r_1\le r_2\le \dots \le r_T$, for which the following holds.
    \begin{itemize}[leftmargin=10px,itemsep=0px]
        \item For every $i\in [T]$, $\big| r_i - \Pr[\calA(S) \text{ outputs $\top$ for the $i$-th query} \mid \text{$\calA(S)$ did not halt before}]\big| \le \eta$.
        \item $r_i$ depends only on the set $S\cap [0, x_i]$. Namely, it is oblivious to data points larger than $x_i$.
    \end{itemize}
\end{lemma}

\begin{proof}
Let $N$ be a large enough integer to be specified later. We construct $x_i,p_i$ in the increasing order of $i$. To start, consider the discrete ensemble of points $P = \{p_i = \frac{i}{N+1}\}_{1\le i\le N} \subseteq [0,1]$.

We set $x_1:= p_1$. Then, consider the $k$-uniform complete hypergraph with vertex set $V = P\setminus \{p_1\}$. Let $S\subseteq V$ be a hyperedge over $k$ vertices. We color the edge with the color
\[
\left\lfloor \frac{1}{\eta}\cdot \Pr[A(S) \text{ outputs $\top$ for the first query}] \right\rfloor \in [0,\frac{1}{\eta}].
\]
As the number of colors is finite, for any desired $N^{(1)} > 1$, there exists $N$ such that there exists a $N^{(1)}$-monochromatic clique. Denote by $P^{(1)}\subseteq P\setminus \{p_1\}$ the subset.

Next, we consider a $(k-1)$-uniform complete hypergraph with vertex set $P^{(1)}$. Let $S\subseteq P^{(1)}$ be a hyperedge over $(k-1)$ vertices. Color the edge with the color
\[
\left\lfloor \frac{1}{\eta}\cdot \Pr[A(S\cup \{x_1\}) \text{ outputs $\top$ for the first query} ] \right\rfloor.
\]
Again, for any desired $N^{(2)}$, there exists $N^{(1)}$ so that the existence of an $N^{(2)}$-monochromatic clique is guaranteed. Denote by $P^{(2)}\subseteq P^{(1)}$ the subset.

Now, note that if we choose $x_1$ to be $p_1$ and choose $x_2,\dots, x_T$ from $P^{(2)}$, then the two conditions in the theorem statement are satisfied for $i = 1$. We proceed to choose $x_2,\dots, x_T$ from $P^{(2)}$ so that the theorem is also satisfied for $i \ge 2$. 
We simply choose $x_2$ as the smallest element in $P^{(2)}$. Then, we apply Ramsey's theory four times, as follows:
\begin{itemize}[leftmargin=20px]
    \item First, We identify a subset $P^{(3)}$ of $P^{(2)}$ such that
    \[
    \Pr[\calA(S\cup \{x_1,x_2\}) \text{ outputs $\top$ for the $2$nd query }]
    \]
    is almost the same (i.e., up to an error of $\eta$) for every $S\subseteq P^{(2)}$.
    \item Continuing, we apply Ramsey's theory three more times to identify $P^{(6)}\subseteq P^{(5)}\subseteq P^{(4)}\subseteq P^{(3)}$ such that, for every data set $S\subseteq P^{(6)} \cup \{x_1,x_2\}$, the quantity
    \[
    \Pr[\calA(S) \text{ outputs $\top$ for the $2$nd query }]
    \]
    only depends on $S\cap \{x_1,x_2\}$ up to an error of $\eta$.

    In total, we need to apply Ramsey's theory four times because the intersection $S\cap \{x_1,x_2\}$ for a data set $S\subseteq P^{(6)}\cup \{x_1,x_2\}$ has four possibilities.
\end{itemize}

After this step, we have identified $x_1,x_2$ and a set $P^{(6)}$ such that the theorem statements are satisfied for $i=1,2$ as long as we choose $x_3,\dots, x_{T}$ from $P^{(6)}$. By setting $N$ to be large enough, we can make $P^{(6)}$ arbitrarily large. We apply the same idea in the following stages. Namely, at each stage $i\in [T]$, we choose $x_i$ as the smallest element in the current set of ``active points'' $P^{(t_i)}$. Then, we apply Ramsey's theory enough times to identify $P^{(t_{i+1})}$ such that the theorem is satisfied for the $i$-th query as long as we choose $x_{i+1}\dots, x_{T}$ from $P^{(t_{i+1})}$. Once we have chosen $x_T$, the construction is complete.

If we start with a sufficiently large $N$ and point set $P = \{\frac{i}{N+1}\}_{i\in [N]}$, we can guarantee the success of each application of Ramsey's theory. This completes the proof.
\end{proof}

\subsubsection{The Reduction}

We are ready to finish the reduction.

\begin{theorem}\label{theo:ramsey-reduction}
Let $\calA$ be an $(\varepsilon,\delta)$-DP algorithm for the online threshold algorithm that operates on $k$ data points and has failure probability at most $\beta$. For every $\eta > 0$, let $\{x_1,\dots, x_T\}\subseteq [0,1]$ be as given in \Cref{lemma:apply-ramsey} with parameter $\eta$. Then, there exists an algorithm $\calB$ for the online Counter Problem that is $(\varepsilon,\delta+ T\eta)$-DP and has failure probability $T\eta + \beta$.
\end{theorem}

Note that the lower bound for the threshold query problem (i.e., \Cref{thm:threshld-query}) follows by combining \Cref{theo:ramsey-reduction} and \Cref{monitorLB:thm}.

\begin{proof}
We describe the algorithm $\calB$.
\begin{itemize}
    \item $\calB$ maintains a subset $S\subseteq \{x_1,\dots, x_T\}$, which we think of as the ``input'' to $\calA$.
    \item For each $i\in [T]$, $\calB$ receives the $i$-th update $\Delta_i$. If $\Delta_i = 1$, $\calB$ inserts $x_i$ to $S$. Otherwise, $\calB$ does not update $S$ in this round. Then, $\calB$ invokes $\calA$ with the partial data set $S$ and query $x_i$. By Theorem~\ref{lemma:apply-ramsey}, knowing only $S\subseteq \{x_1,\dots, x_i\}$, there is $r_i\in [0,1]$ that captures the halting probability of $\calA$ up to an error of $\eta$. $\calB$ simply returns $\top$ with probability $r_i$. This finishes the description of the algorithm $\calB$.
\end{itemize}
For the correctness, note that for any update sequence $\Delta = (\Delta_1,\dots, \Delta_T)$ of Hamming weight at most $k$, it induces a data set $S = \{x_i : \Delta_i = 1\}$. Furthermore, the output distribution of $\calB$ on $\Delta$, and the output of $\calA$ on $S$ (and queries $x_1,\dots, x_T$) differ by at most $\eta\cdot T$ in statistical distance. Furthermore, the correctness criteria of the two problems are equivalent. Hence, if $\calA$ has failure probability $\beta$, we know that $\calB$ has failure probability at most $\eta \cdot T + \beta$.

For the privacy claim, simply note that two adjacent update sequences $\Delta, \Delta'$ induce two adjacent data sets $S,S'$. Thus, the privacy of $\calB$ follows from the privacy of $\calA$ and the similarity of output distributions between $\calA$ and $\calB$.
\end{proof}

\newpage

\bibliography{mybib}

\newpage
\appendix

\section{Lower Bound for Private Online Learning} \label{sec:lowerboundOnlineLearning}

In this section we prove Theorem~\ref{thm:onlineLearning}. 
Recall that this theorem states that every private online learner for the class of all point functions over a somain $\XXX$ must make $\Omega(\log T)$ mistakes with constant probability. We first give a proof sketch for a simplified lower bound that holds only for {\em proper} learners, i.e., for learners that on every time step $i$ choose a hypothesis $h_i$ that is itself a point function. 
For simplicity, in this intuitive explanation we will assume that $|\XXX|\gg T$, say $|\XXX|=T^2$. Suppose towards contradiction that there exists a private online {\em proper} learner $\AAA$ for the class of all point functions with error $k\ll\log T$. We now use $\AAA$ in order to construct an algorithm $\MMM$ that solves the threshold monitor problem (Problem~\ref{threshmon:problem}), which will contradict Theorem~\ref{monitorLB:thm}. At the beginning of the execution, algorithm $\MMM$ chooses a random point $x^*\in\XXX\setminus0$. Then, on every time step $i$ we do the following: 
\begin{enumerate}
    \item Algorithm $\MMM$ obtains an input bit $b_i\in\{0,1\}$.
    \item If $b_i=0$ then set $x_i=0$. Otherwise set $x_i=x^*$.
    \item Feed the input $x_i$ to algorithm $\AAA$, obtain a hypothesis $h_i$, and then feed it the ``true'' label $b_i$.
    \item If $h_i(x^*)=0$ then output $\bot$ and continue to the next step. Otherwise output $\top$ and halt.
\end{enumerate}

Now, whenever $\AAA$ makes at most $k$ mistakes, then algorithm $\MMM$ halts {\em before} the $(k+1)$th time in which the input bit is 1. On the other hand, observe that before the first time that we get a positive input bit, then algorithm $\AAA$ gets no information on $x^*$, and hence the probability that $h_i(x^*)=1$ is very small (at most $1/|\XXX|$). This is because when $\AAA$ is a proper learner, then $h_i$ can label at most 1 point from $\XXX$ as 1, and $x^*$ is chosen at random. Hence, when $|\XXX|$ is large enough, we get that w.h.p.\ algorithm $\MMM$ never halts before the first time that we get a positive input bit. This contradicts our impossibility result for the threshold monitor problem.\footnote{Technically, in the definition of the threshold monitor problem we required the algorithm to output $\bot$ whenever the number of ones in the stream is at most $k/2$, while here we allowed the algorithm to output $\top$ immediately after the first time that the stream contains a positive input. This is a minor issue which we ignore for the sake of this simplified explanation. }

We now proceed with the formal proof, which is more complex in order to capture also the improper case. Intuitively, the extra difficulty is that in the improper case the algorithm might choose hypotheses that label the point $x^*$ as 1 even before seeing any instance of this point.

\medskip
\begin{proof}{\bf of Theorem~\ref{thm:onlineLearning}}\;\;
    Suppose for the sake of contradiction that such a learner exists and denote it by $\mathcal{A}$.

    Denote $k = \sqrt{T}$, 
    and consider the following distribution over query sequences.  
    \begin{itemize}
        \item For the $i$-th query when $i\in [1,k]$, send $(x_i,0)$ where $x_i\sim \Uniform[1,k]$ (where $\Uniform[a,b]$ is the uniform distribution on $\{a,\ldots,b\}$).
        \item For the $i$th query when $i\in \{k+1,\ldots,2k\}$, send $(x_i, 0)$ where $x_i\sim \Uniform[k+1, 2k]$.
        \item In general for $j\in [k]$, for the $i$th query where $i\in \{(j-1)k+1, \ldots, jk\}$, send $(x_i, 0)$ where $x_i\sim \Uniform[(j-1)k,jk]$ 
    \end{itemize}
We use $k$ phases of $k$ queries each.
    For each phase $j\in [k]$ there is an allocated disjoint set of $k$ points $X_j := \{(j-1)k+1,\ldots,jk\}$. The query points at each step of phase $j$ is selected uniformly at random from $X_j$.
    
    Under this query distribution, for every $i,j$ we denote
    $p_i^{(j)} := \Pr[h_i(j) = 1],$ where $h_i$ is the hypothesis released at the $i$-th step by  $\mathcal{A}$.
    Note that the events $\mathbf{1}\{h_i(j) = 1\}$ depend on the randomness of the query sequence and the algorithm and are not necessarily independent.

Note that the query distribution we defined always chooses the label 0. This might seem strange at first glance. What we show, however, is that a private online learner cannot keep releasing the all-zero hypothesis (whereas a non-private learner would naturally do this with this query distribution). That is, a private learner must sometimes ``gamble'' and output a hypothesis that is not all zero. But these ``gambles'' (as we show) will inflate the mistake bound of the non-private learner and yield the desired lower bound. Formally, the analysis proceeds by inspecting the following two cases:

\paragraph*{Case 1.} For every phase $j\in [k]$ and every $k'\in [(j-1)k+1,jk]$, we have
\[
\sum_{i=(j-1)k+1}^{jk} p_i^{(k')} \ge 0.001,
\]
that is, the average probability, over steps $i$ in the phase, that the released hypothesis is $1$ on $k'$  is at least $0.001/k$. 

We can now express the expected number of mistakes by the algorithm:
\[
\sum_{j=1}^{k} \sum_{i=(j-1)k+1}^{jk} \frac{1}{k} \sum_{k'\in X_j} p_i^{(k')} \ge \Omega(\sqrt{T}).
\]

\paragraph*{Case 2.} There exists a phase $j\in [k]$ and point $k' \in X_j$ such that
\begin{equation} \label{case2cond:eq}
 \sum_{i=(j-1)k+1}^{jk} p_i^{(k')} \le 0.001.
\end{equation}

In this case, we first derandomize the query distribution. 

Consider the query sequence and the event that it does not include examples of the form $(k', 0)$, that is, for all steps $i$ in phase $j$, $x_i \not= k'$.
This event has probability 
$(1-\frac{1}{k})^{k} \geq 1/4$.  
Therefore, from~\eqref{case2cond:eq} it must hold that when conditioning on the event that $k'$ is not present we have
\begin{equation}
    \sum_{i=(j-1)k+1}^{jk} \Pr[h_i(k') = 1] \le 0.001/0.25 \leq 0.01.
\end{equation}
This implies that with probability at least $0.99$ over sequences conditioned on the event, it holds that 
$h_i(k') = 0$ for all of the steps $i\in [(j-1)k+1,jk]$.

Therefore we can fix a sequence for which $k'$ is not present and with probability 0.99, $h_i(k')=0$ on all steps.

We now apply the construction of Algorithm~\ref{algo:construction} to insert $\Omega(\log(k)) = \Omega(\log(T))$ queries of the form $(k',1)$ into the sequence, such that the probability that $h_i(k') = 0$ throughout is still $\Omega(1)$. This in particular means that the algorithm makes $\Omega(\log(T))$ mistakes on the new query sequence, as desired.
\end{proof}

\section{JDP Algorithm for the Mirror Problem} \label{sec:JDPmirror}

\begin{algorithm2e}[t] \label{algo:JDPmirror}
    \caption{\texttt{JDP-Mirror}}
    \DontPrintSemicolon 
    \KwIn{Privacy parameters $\eps\in (0,2],\delta\in(0,0.1)$, Delay parameter $K\geq 1$}
    \tcp{Initialization}
    $\eps' \gets O(\eps^2/\log(1/\delta))$;
    $\delta' \gets O(\delta) $\tcp*{As in Lemma~\ref{JDPMirror_privacy:lemma}} $C \gets K$\;
    \For(\tcp*[f]{Main Loop}){$t=1,\ldots$}
    {
      Receive an input bit $\Delta_t\in \{0,1\}$\;
      \eIf(\tcp*[f]{Process input $1$}){$\Delta_t=1$}{$C \gets C+1$\;
      $r \gets \Bern(\pi_{\eps',\delta'}(\max\{0,C-K\}))$\tcp*{Where $(\pi(i))_i$ are as in Claim~\ref{pi:claim}}
      \eIf{$r=1$}{Output $y_t \equiv \top$} 
      {Output $y_t \equiv \bot$}
      }{Output $y_t \equiv \bot$ \tcp*{Input is $\Delta_t=0$}}
}
\end{algorithm2e}

In this section we present a construction for (a variant of) the threshold monitor problem that satisfies a weaker variant of differential privacy known as {\em joint differential privacy (JDP)} and has an error that is independent of $T$. 
We will use this construction in Section~\ref{JDPpoint:sec} as a building block of an online prediction algorithm for point functions.
We first recall the definition of JDP.

\begin{definition} [$(\eps,\delta)$-indistinguishable]
For two random variables $X_1,X_2$ 
we write $X_1 \approx_{\eps,\delta} X_2$ if 
 for $b\in \{0,1\}$
\[
\max_{S\subseteq\supp(X^{b})} \ln\left[\frac{\Pr[X^b\in S]-\delta}{\Pr[X^{1-b}\in S]} \right] \leq \eps.
\]
We will interchangeably use notation for random variables and distributions.
\end{definition}

\begin{definition}[Joint Differential Privacy \cite{KearnsPRRU15}]
Let $\MMM:X^n\rightarrow Y^n$ be a randomized algorithm that takes a dataset $S\in X^n$ and outputs a vector $\vec{y}\in Y^n$. Algorithm $\MMM$ satisfies $(\eps,\delta)$-joint differential privacy (JDP) if for every $i\in[n]$ and every two datasets $S,S'\in X^n$ differing only on their $i$th point it holds that 
$\MMM(S)_{-i}\approx_{\eps,\delta}\MMM(S')_{-i}$. Here $\MMM(S)_{-i}$ denotes the (random) vector of length $n-1$ obtained by running $(y_1,\dots,y_n)\leftarrow \MMM(S)$ and returning $(y_1,\dots,y_{i-1},y_{i+1},\dots,y_n)$.
\end{definition}

\begin{problem}[The \texttt{Mirror} Problem]
For a sequence of input bits $(\Delta_t)_{t\in [T]}$ where $\Delta_t\in\{0,1\}$, the goal is to mirror the input, returning $y_t=\bot$ for $\Delta_t=0$ and returning $y_t=\top$ for $\Delta_t=1$. 
The constraints are that for a specified delay parameter $K\geq 1$:
\begin{itemize}
\item
   For $\Delta_t=0$ we must return $y_t=\bot$.
\item    
    We may return $y_t=\top$ only if $\sum_{i=1}^{t-1} \Delta_t \geq K$, that is,    
    there were at least $K$ prior steps with input $\Delta_t=1$. 
\end{itemize}
The goal is to minimize the number of \emph{mistakes}. A mistake occurs at time $t$ if
$\Delta_t=1$ and  $\sum_{i=1}^t \Delta_t > K$, and
$y_t = \bot$.
\end{problem}

For privacy parameters $(\eps,\delta)$ and delay parameter $K\geq 1$ we construct an $(\eps,\delta)$-JDP
algorithm for the \texttt{Mirror} problem.  The construction is described in 
Algorithm~\ref{algo:JDPmirror}.
The algorithm uses the probabilities specified in the following claim:
\begin{claim}\label{pi:claim}
For $\eps>0$ , $\delta\in (0,0.1)$ there are probabilities 
$\pi_{\eps,\delta}(i)$ ($i\geq 0$) and $L \leq \lceil \frac{1}{\eps}\ln((e^\eps -1)/(2\delta)) \rceil$ as follows:
\begin{itemize}
\item $\pi_{\eps,\delta}(0) = 0$, $\pi_{\eps,\delta}(1) = \delta$, $\pi_{\eps,\delta}(2L) = 1-\delta$
\item For $i>2L$, $\pi_{\eps,\delta}(i) = 1$
\item For $1\leq i < 2L$, $\Bern[\pi_{\eps,\delta}(i)] \approx_\eps \Bern[\pi_{\eps,\delta}(i+1)]$.
\end{itemize}
\end{claim}
\begin{proof}
    Let
    \begin{align*}
       L = \arg\min j \ , \delta \sum_{i=0}^{j-1} e^{i\eps} > 1/2\ . 
    \end{align*}
    Let $\eps'\leq \eps$ be the solution of
$\delta \sum_{i=0}^{L-1} e^{i\eps'} = 1/2$.
    Define
    \begin{align*}
        \text{For $1\leq\i \leq L$: } & \pi(i) = \delta e^{(i+1) \eps'} \\
        \text{For $L+1\leq\i \leq 2L$: } & \pi(i) = \delta e^{(2L-i) \eps'} \ .
    \end{align*}
    
\end{proof}

\begin{lemma} [Utility of \texttt{JDP-Mirror}]
 Algorithm~\ref{algo:JDPmirror} correctly solves the \texttt{Mirror} problem for delay parameter $K$ while incurring at most $2L= O(\eps^{-2}\log^2(1/\delta))$ mistakes.
\end{lemma}
\begin{proof}
    On time steps where $\Delta_t=0$ the algorithm always returns $y_t=\bot$, so there are no mistakes. On the first $K$ time steps with $\Delta_t=1$ we have $C\leq K$ and thus $r= 0$ and hence $y_t=\bot$, which is as specified. At time steps where $C>K+2L$ and $\Delta_t=1$ we have $r=1$ and the algorithm always reports $y_t=\top$ so there are no mistakes.
    Mistakes are possible only when $\Delta_t=1$ and only when $K+1\leq C\leq K+2L$ so there are at most $2L = O((1/\eps')\log(1/\delta'))$ mistakes.  
\end{proof}

\begin{lemma} [JDP of \texttt{JDP-Mirror}] \label{JDPMirror_privacy:lemma}
 Algorithm~\ref{algo:JDPmirror} is $(\eps,\delta)$-JDP.
\end{lemma}
\begin{proof}
Fix $i^*\in[T]$ and fix two neighboring input sequences 
$\vec{\Delta},\vec{\Delta}'\in\{0,1\}^T$
that differ in position $i^*$. 
Observe that coordinates $t$ s.t.\ $\Delta_t=0$ (and hence when $t\not=i^*$, $\Delta'_t=0$) have no effect on either of the executions. Hence, we may ignore these coordinates. Therefore, without loss of generality, we may assume that 
$\vec{\Delta}\equiv \boldsymbol{1}$, and that $\vec{\Delta}'$ is the vector of all 1 except for the $i^*$th coordinate for which $\Delta'_{i^*}=0$. 

Consider now for each $t\in [T]\setminus\{i^*\}$ the random variables $Y_t$ that is the output at time $t$ when the input is $\vec{\Delta}$ and the random variable $Y'_t$ that is the output at time $t$ when the input is $\vec{\Delta}'$. Note that for a fixed input, the random variables of different time steps are independent.
Denote $P_t = \Pr[Y_t = \top]$ and $P'_t = \Pr[Y'_t = \top]$.  
We will show that there can be at most $2L$ time steps  $t\in [T]$ where $P_t \not= P'_t$ and that  
out of these steps, there can be at most two where $|P_t-P'_t|\leq \delta'$ (and thus $Y_t\approx_{0,\delta'}Y'_t$) and at most $2L-2$ 
where $Y_t\approx_{\eps'}Y'_t$.
If we apply advanced composition over time steps $t\in [T]\setminus\{i^*\}$
we obtain that
$\boldsymbol{Y}_{-i^*} \approx_{\eps,\delta} \boldsymbol{Y}_{-i^*}$ where
\begin{align*}
\eps &= \frac{1}{2} (2L-2) (\eps')^2 + \eps' \sqrt{2(2L-2) \log(1/\delta'')} = O(\sqrt{\eps'} \log(1/\delta''))\\
\delta &= 2\delta'+\delta''\ .
\end{align*}
Therefore if we wish for final privacy parameters $\eps,\delta$ we need to choose 
$\eps' = O((\eps)^2 \log^{-2}(1/\delta))$ and $\delta' = O(\delta)$.

We can tighten the privacy analysis using the Target Charging Technique~\citep{TCT:Neurips2023}.  We note that the $2T-2$ steps that can differ between the dataset have $\top$ probabilities $\pi_{\eps',\delta'}(j)$ for $j\in[2L]$. With TCT privacy analysis we only count \emph{target hits}. The expected number of target hits is
$\sum_{j\in [2L]} \min\{\pi_{\eps',\delta'}(j), 1-\pi_{\eps',\delta'}(j) \} = O(1/\eps')$. We obtain a composition over $O(\max\{1/\eps',\log(1/\delta)\})$ steps (the $\log(1/\delta)$ is the overhead of TCT). Applying advanced composition and solving for $\eps'$ we obtain:
$\eps' = O((\eps)^2 \log^{-1}(1/\delta))$ and $\delta' = O(\delta)$.

It remains to establish that the claimed property holds. We do so via case analysis. Denote by
$C_t = \sum_{i=1}^t \Delta_i$ and $C'_t= \sum_{i=1}^t \Delta'_i$ the value of the counter $C$ in Algorithm~\ref{algo:JDPmirror} at step $t$ for inputs $\boldsymbol{\Delta}$ and $\boldsymbol{\Delta'}$ respectively. Observe that we always have $C_t-1\leq C'_t \leq C_t$.
\paragraph{Case $i^* \leq K+1$}
\begin{itemize}    
\item
For $t\leq K$, $C'_t\leq C_t\leq K$ and hence $P_t=P'_t=\pi_{(\eps',\delta')}(0)=0$. 
\item 
For $t=K+1$ (relevant when $i^*<K+1$), $C_t = K+1$ and hence $P_t = \pi_{(\eps',\delta')}(C_t-K)= \pi_{(\eps',\delta')}(1)= \delta'$ and $C'_t=K$ and hence $P'_t=\pi_{(\eps',\delta')}(0)=0$.
\item 
For $K+1 < t \leq K+2L$, $K+2L \geq C_t = C'_t+1 \geq K+2$
That is $2\leq C_t-K \leq 2L$ and $1\leq C'_t-K \leq 2L-1$ and hence by Claim~\ref{pi:claim}, $Y_t \approx_{\eps'} Y'_t$. 
\item
For $t=K+2L+1$, we have $C_t=K+2L+1$ and $C'_t = K+2L$ and hence $P_t=1$ and $P'_t = 1-\delta'$.
\item 
 For $t\geq K+2L+1$, $P_t=P'_t=1$.
\end{itemize}

\paragraph{Case $K+1 < i^* \leq K+2L+1$}
\begin{itemize}
\item
For $t<i^*$,  $P_t = P'_t$.
     \item
For $i^*<t \leq K+2L$, we have $K+2 \leq C_t = C'_t+1 \leq K+2L$ and $2\leq C_t-K\leq K+2L$ and hence by Claim~\ref{pi:claim}, $Y_t \approx_{\eps'} Y'_t$. 
  \item 
 For $t= K+2L+ 1$ (relevant when $i^*< K+2L+1$), we have $C_t = K+2L+1$ and hence $P_t=1$ and $C'_t= K+2L$ and hence $P'_t=1-\delta'$.
 \item
If $t > K+2L+1$, then $C'_t-K > 2L$ and $C_t-K > 2L+1$ and therefore $P_t=P'_t=1$.
\end{itemize}

\paragraph{Case $i^* >  K+2L+1$}
In this case for all $t\leq K+2L+1$, $C_t=C'_t$ and hence $P_t=P'_t$ and for all $t> K+2L+1$,  $C_t-K > 2L+1$ and $C'_t-K > 2L$ and hence $P_t=P'_t=1$.

Summarizing all this, in each of the cases we have a composition of at most $2L-2$ $\eps'$-indistinguishable random variables and at most two $(0,\delta')$-indistinguishable random variables.
\end{proof}

\section{Private Online Prediction for Point Functions} \label{JDPpoint:sec}

In this section we present a construction for a private online predictor for point functions, obtaining mistake bound that is independent of the time horizon $T$. This separates private online learning from private online prediction. 
Specifically, we leverage Algorithm~\ref{algo:JDPmirror} to construct an efficient online learner for point functions in the private prediction model. The construction is described in Algorithm~\ref{algo:onlinepredictor}.

\begin{algorithm2e}[t] \label{algo:onlinepredictor}
    \caption{Online Predictor for Point Functions}
    \DontPrintSemicolon 

    \begin{enumerate}[itemsep=0px]
    \item Instantiate \texttt{JDP-Mirror} (Algorithm~\ref{algo:JDPmirror}) with privacy parameters $\eps,\delta$ and with the delay $K=20 \cdot k$, where $k=\Theta(\frac{1}{\eps}\log\frac{1}{\delta})$.
    
    \item Set ${\rm Flag}=0$ and ${\rm Count}=0$.
        \item For time $t=1,2,\dots,T$ do:
\begin{enumerate}
    \item Obtain an input point $x_t$.
    \item If ${\rm Flag}=0$ then:
    \begin{enumerate}
        \item Output 0 and obtain the ``true'' label $y_t$.
        \item Set ${\rm Count}\leftarrow{\rm Count}+y_t$.
        \item If ${\rm Count}\geq k$ then
        \begin{enumerate}

            \item\label{step:histo} Run an $(\eps,\delta)$-DP algorithm for sparse histograms to identify a point $x^*\in\XXX$ that appears $\Omega(\frac{1}{\eps}\log(1/\delta))$ times in the multiset $\{ x_j : j\in[t] \text{ and } y_j=1\}$. If the algorithm failed to identify such a point then we set $x^*\gets \star$ for some special symbol $\star\notin\XXX$. See, e.g., \cite{sparsehistograms:ITCS2016,CohenGeriSarlosStemmer:ICML2021}.

            \item\label{step:mirrorPast} For $j\in[t]$, feed \texttt{JDP-Mirror} the bit $I\left(x_j,x^*\right)$, where $I(a,b)=1$ iff $a=b$.
            \item\label{step:Flag} If all the answers given by \texttt{JDP-Mirror} in the previous step were $\bot$ then set ${\rm Flag}=1$. Otherwise set ${\rm Flag}=3$.

        \end{enumerate}
    \end{enumerate}
    \item Else if ${\rm Flag}=1$ or ${\rm Flag}=2$ then:
    \begin{enumerate}
        
    \item Feed \texttt{JDP-Mirror} the bit $I(x_t,x^*)$ to obtain $a_t\in\{\top,\bot\}$. 

    \item If ${\rm Flag=1}$ and $a_t=\top$ then
    \begin{itemize}

        \item\label{step:L} Let $L\subseteq[t]$, where $|L|\leq10k$, be the set containing the positions of the {\em first} $10k$ occurrences of $x^*$ till time $t$.
        
        \item\label{step:fake} Let ${\rm Fake}\leftarrow |\{j\in L : y_j=0\}|+\Lap(\frac{1}{\eps})$. 
                \item If ${\rm Fake}\geq\frac{1}{\eps}\log\frac{1}{\delta}$ then set ${\rm Flag}=3$ and output 0. 
                \item Otherwise set ${\rm Flag}=2$ and output the bit $I(a_t,\top)$.
    \end{itemize}
    
    \item Else output the bit $I(a_t,\top)$.

    \item Obtain the ``true'' label $y_t$. 

\end{enumerate}

    \item Else if ${\rm Flag}=3$ then output 0.

    \end{enumerate}
    \end{enumerate}
\end{algorithm2e}

\begin{lemma}\label{lem:pointLearningEps2} [Restatement of Lemma~\ref{lem:onlinePredictionPoints}]
    Let $\mathcal{X}$ be a domain. The class of point functions  over $\mathcal{X}$ admits an $(\eps, \delta)$-DP online predictor with mistake bound $O(\frac{1}{\eps^2}\log^2\frac{1}{\delta})$.
\end{lemma}

\begin{proof}{\bf sketch}\;\;
The proof is via the construction given in Algorithm~\ref{algo:onlinepredictor}, which we denote here as algorithm $\AAA$. The utility analysis of this algorithm is straightforward. We only need to consider the case where the input is consistent with some point function, say with point $x^*$. In this case, we argue that the output is identical to that of \texttt{JDP-Mirror} with delay $K$ on inputs $\mathbf{1}(x_i=x^*)$, with $\top$ replaced by $1$ and $\bot$ replaced by $0$. Once $\approx\frac{1}{\eps}\log\frac{1}{\delta}$ positive labels have been obtained, we run the private histogram algorithm in order to identify the point $x^*$. Since the delay $K$ is much larger than $k$ we have $\text{Flag}=1$. (Note that if \texttt{JDP-Mirror} would have returned a $\top$ before that point, we know that the input is not consistent with a point function.) After \texttt{JDP-Mirror} returns the first $\top$ we get that $\text{Fake}=\Lap(\frac{1}{\eps})$ and the algorithm transitions to $\text{Flag}=2$. The output of the algorithm then follows \texttt{JDP-Mirror} and the number of mistakes that is at most $K+O(\eps^{-2}\log^2(1/\delta))= O(\eps^{-2}\log^2(1/\delta))$.

For the privacy analysis, fix two neighboring input streams that differ on the $i$th labeled example: 
\begin{align*}
    S_0&=((x_1,y,1),\dots,(x^0_i,y^0_i),\dots,(x_T,y_T))\\ S_1&=((x_1,y,1),\dots,(x^1_i,y^1_i),\dots,(x_T,y_T)) . 
\end{align*}

Let $b\in\{0,1\}$ and consider the execution of $\AAA$ on $S_b$. 
We now claim that the outcome distribution of $\AAA(S_b)$, excluding the $i$th output, can be simulated from the outcomes of three private computations on $b$.
Specifically, we consider a simulator $\Sim$ that interacts with a ``helping algorithm'', $\Helper$, who knows the value of the bit $b$ and answers DP queries w.r.t.\ $b$ as issued by $\Sim$.
The simulator $\Sim$  
first asks $\Helper$ to execute
the algorithm for histograms on the dataset containing the first $k$ positively labeled points in $S_b$ to obtain $x^*$. This satisfies $(\eps,\delta)$-DP.  The simulator then asks $\Helper$ for the outcome vector of executing \texttt{JDP-Mirror} on the stream $\left(I(x_j,x^*)\right)_{j=1}^{T}$. The $\Helper$ responds with the vector $\vec{b}\in\{\bot,\top,\emptyset\}^T$ (excluding the $i$th step). 
Next, similarly to Step~\ref{step:L} of the algorithm, let $\hat{L}$ denote the indices of the first $10k$ occurrences of $x^*$ in $S_b$. The simulator asks $\Helper$ for the value of 
$$
\widehat{{\rm Fake}}\leftarrow |\{j\in\hat{L}:y_j=0\}|+\Lap(\frac{1}{\eps}).
$$
This satisfies $(\eps,0)$-DP by the properties of the Laplace mechanism.

We now claim that using these three random elements ($x^*$, the outcome vector of \texttt{JDP-Mirror}, and $\widehat{{\rm Fake}}$), we can simulate the outcome of $\AAA(S_b)$ up to a statistical distance of at most $\delta$. Specifically, if $\widehat{{\rm Fake}}\geq\frac{1}{\eps}\log\frac{1}{\delta}$ then we output the all 0 vector, and otherwise we output the vector given by \texttt{JDP-Mirror} (where $\top$ is replaced by 1 and $\bot$ is replaced by 0). To see that this simulates the outcome of $\AAA(S_b)$, let us consider the following two events:
\begin{enumerate}
    \item The first $\top$ given by \texttt{JDP-Mirror} is such that the number of appearances of $x^*$ before it is at least $10k$. By the properties of \texttt{JDP-Mirror}, this event happens with probability at least $1-\delta$.
    \item The noise magnitude in $\widehat{{\rm Fake}}$ is at most $\frac{1}{\eps}\log\frac{1}{\delta}$. This event also holds with probability at least $1-\delta$ by the properties of the Laplace distribution.
\end{enumerate}

Now, whenever these two events hold, then the outcome of the simulator is identical to the outcome of $\AAA(S_b)$. More specifically, if the algorithm sets ${\rm Flag}=1$ in Step~\ref{step:Flag}, then the outcome of the simulator is syntactically identical to the output of the algorithm. Otherwise (if we set ${\rm Flag}=0$ in Step~\ref{step:Flag}), then there must be at least $9k$ occurrences of $(x^*,0)$ in $S_b$ before the $k$th occurrence of $(x^*,1)$. Hence, in this case we have that $\widehat{{\rm Fake}}\geq\frac{1}{\eps}\log\frac{1}{\delta}$. So both the algorithm and the simulator output the all zero vector.

Overall, we showed that by posing 3 queries to $\Helper$, each of which satisfies $(\eps,\delta)$-DP, our simulator can simulate the outcome of $\AAA(S_b)$ up to a statistical distance of at most $\delta$. This concludes the proof.
\end{proof}

\end{document}